\documentclass[10pt]{article}
\usepackage[a4paper]{geometry}
\usepackage[utf8]{inputenc}
\usepackage[T1]{fontenc}
\usepackage{setspace,authblk}
\usepackage{amsmath,mathrsfs,amsfonts,amsthm,amssymb}

\usepackage[english]{babel}

\usepackage{graphicx}
 
\usepackage{hyperref}
\title{\singlespace {On the state-space model of unawareness}}

\author[]{Alex A.T. Rathke\thanks{NECCT/FEA-RP/USP, University of S\~ao Paulo. \texttt{alex.rathke@alumni.usp.br}}}

\date{\today}

\theoremstyle{plain}
\newtheorem{theorem}{Theorem}
\newtheorem*{theorem*}{Theorem}

\newtheorem*{proposition*}{Proposition}
\newtheorem{remark}{Remark}
\newtheorem*{remark*}{Remark}

\newtheorem*{condition*}{Condition}
\newtheorem{definition}{Definition}
\newtheorem*{definition*}{Definition}

\newtheorem*{corollary*}{Corollary}

\begin{document}

\maketitle

\begin{abstract} %abstract here

We show that the knowledge of an agent carrying non-trivial unawareness violates the standard property of 'necessitation', therefore necessitation cannot be used to refute the standard state-space model. A revised version of necessitation preserves non-trivial unawareness and solves the classical Dekel-Lipman-Rustichini result. We propose a generalised knowledge operator consistent with the standard state-space model of unawareness, including the model of infinite state-space.

\end{abstract}

\noindent\textbf{Keywords:} Knowledge, Unawareness, State-space model, Plausibility, KU introspection, AU introspection, Negative introspection, R necessitation, Symmetry.
\\
\noindent\textbf{JEL Classification:} C70, C72, D80, D82.

\section{Model and results} \label{Model and results}

Recall the standard state-space model of knowledge and unawareness, where we have a finite set of possible states $\Omega$ and a possibility correspondence $P : \Omega \rightarrow 2^{\Omega}$ which maps states $s \in \Omega$ to subsets of $\Omega$. At some state $s$, the set $P(s)$ is the set of states which the agent considers to be possible, i.e. representing the information processing capacity of the agent \cite{geanakoplos1989}, or his epistemic state \cite{bacharach1985}. For some event $E \subseteq \Omega$, we say that the agent with a possibility correspondence $P$ knows $E$ at the state $s$ if $P(s) \subseteq E$. The standard knowledge operator is defined by

\begin{equation} \label{k01}
K(E) := \{ s \in \Omega | P(s) \subseteq E \} = KE . \\
\end{equation}

The set $KE$ is the set of states in which the agent knows that $E$ must have occurred, so $KE \neq \emptyset$ means 'the agent knows $E$'. The set $\neg KE = \Omega \setminus KE$ is the set of states in which the agent does not know $E$. Iterations of the knowledge operator $K$ produce higher-order knowledge, for the agent becomes aware of her knowledge. The intuition is that if the agent knows $E$, than she knows that she knows $E$, than she knows that she knows that she knows $E$, and so on. 

On the other hand, the set of states where the agent is fully unaware of $E$ implies that she does not know $E$, and she does not know that she does not know $E$, and she does not know that she does not know that she does not know $E$, and so on. Iterations of the complement of the knowledge $\neg K$ as defined in Eq. \ref{k01} over the event $E$ defines unawareness $U(E)$, for we have %set-builder notation

\begin{equation} \label{u01}
\begin{array}{rl}
U(E) &:= \neg KE \cap \neg K \neg KE \cap \neg K \neg K \neg KE \dots \\
\\
&= \bigcap_{i = 1}^{\infty} (\neg K)^{i} (E) = UE . \\
\end{array}
\end{equation}

\noindent where $(\neg K)^{i}$ indicates that we iterate the $\neg K$ operator $i$ times\footnote{Some authors refer to Eq. \ref{u01} as 'strong plausibility' \cite{heifetz2006,li2009,galanis2013}. Some authors define different levels of unawareness, following from the number of iterations of $\neg K$ within Eq. \ref{u01}, \cite{fukuda2021}.}. Literature on knowledge and unawareness proposes several properties for the $K$ and $U$ operators, see \cite{bacharach1985,geanakoplos1989,modica1994,dekel1998,samuelson2004,heifetz2006,li2009,chen2012,schipper2014,fukuda2021}. 

The agent carries a non-trivial unawareness if there exist some event $E$ such that $UE \neq \emptyset$ \cite{dekel1998,li2009,chen2012,fukuda2021}, i.e. there is some state $s \in \Omega$ such that the agent is unaware of some event $E$\footnote{Non-trivial unawareness derives from a non-partitional possibility correspondence, such that for some state $s$ and some event $E$, $P(s) \neq \emptyset, s \notin KE, UE \neq \emptyset$.}.

\begin{remark} \label{r01}
Non-trivial unawareness $UE \neq \emptyset$ violates the property of 'negative introspection', therefore $UE \neq \emptyset \rightarrow \neg KE \not\subset K \neg KE$ for any $E \subseteq \Omega$ \cite{modica1994,dekel1998}.
\end{remark}

Negative introspection is defined as $\neg KE \subseteq K \neg KE$, and it means that if the agent does not know the event $E$, then she knows that she does not know $E$. The studies of \cite{modica1994,dekel1998} show that negative introspection breaks down with $UE \neq \emptyset$. Intuitively, while the agent may know that there are events she is unaware of, she cannot know the exact events she does not know, therefore $UE \neq \emptyset \rightarrow \neg KE \not\subset K \neg KE$. Otherwise, negative introspection implies that the agent has a partitional possibility correspondence $P$ and she is aware of everything, $UE = \emptyset$ \cite{dekel1998,schipper2014}. 

We apply Remark \ref{r01} for one main result.

\begin{theorem} \label{t01}
For any knowledge operator $K^{\prime}$, non-trivial unawareness $UE \neq \emptyset$ violates the property of 'necessitation', therefore $UE \neq \emptyset \rightarrow K^{\prime} \Omega \neq \Omega$.
\end{theorem}

\begin{proof}
For all events $E \subseteq \Omega$, we have $\emptyset \subseteq K^{\prime} E \subseteq \Omega$, hence $\emptyset \subseteq K^{\prime} \Omega \subseteq \Omega$. We have $K^{\prime} \emptyset = \emptyset$ vacuously, therefore $\neg K^{\prime} \emptyset = \neg \emptyset = \Omega$. Now, negative introspection implies $\neg K^{\prime} \emptyset \subseteq K^{\prime} \neg K^{\prime} \emptyset = K^{\prime} \Omega$, therefore $\Omega \subseteq K^{\prime} \Omega$. Since we also have have $K^{\prime} \Omega \subseteq \Omega$, therefore $K^{\prime} \Omega = \Omega$ \cite{bacharach1985}. However, Remark \ref{r01} shows that non-trivial unawareness $UE \neq \emptyset$ violates negative introspection for any $E \subseteq \Omega$. Therefore, $UE \neq \emptyset \rightarrow K^{\prime} \Omega \neq \Omega$.
\end{proof}

Necessitation is defined as $K\Omega = \Omega$, and it is assumed as a fundamental property of knowledge $K$ \cite{heifetz2006,li2009,galanis2013,schipper2014}. Current literature applies necessitation to derive contradictory results with respect to unawareness $U$, e.g. \cite{dekel1998,chen2012,tada2021}, therefore arguing that the standard state-space model of unawareness is inconsistent \cite{schipper2014}. 

Theorem \ref{t01} shows that non-trivial unawareness $UE \neq \emptyset$ violates necessitation, therefore it cannot be used to prove the standard model inconsistent. For an agent with non-trivial unawareness $UE \neq \emptyset$, simply assuming necessitation means that she triggers her negative introspection to learn about all states $s \in \Omega$, i.e. it changes her possibility correspondence $P$ so it becomes partitional, therefore her unawareness vanishes. 

On the analysis of unawareness $U$, Theorem \ref{t01} indicates that we are required to derive a new definition of knowledge for an agent that carries non-trivial unawareness, and which is consistent with the standard definitions of knowledge $KE$ and unawareness $UE$ in Eq. \ref{k01} and \ref{u01} for any event $E \subseteq \Omega$. By now, all we have is a non-empty set of possible states $\Omega$, events $E$ and possibility correspondence $P$, the definitions of knowledge $K$ and unawareness $U$, and the proof that $UE \neq \emptyset \rightarrow K^{\prime} \Omega \neq \Omega$. We apply Eq. \ref{k01} and \ref{u01} on the following result.

\begin{definition} \label{d01}
For all $E \subseteq \Omega$, $K^{\prime} E := KE \setminus UE$. 
\end{definition}

The generalisation of the knowledge operator $K^{\prime}$ in Definition \ref{d01} is trivial for all events $E \subset \Omega$ since $KE \cap UE = \emptyset \rightarrow K^{\prime}E = KE$, nonetheless it solves the conflict between the non-trivial unawareness $UE \neq \emptyset$ and the application of the standard knowledge operator $K$ on the full state-space $\Omega$. The intuition is that the set of states which the agent knows $E$ necessarily excludes the states which she is unaware of $E$, so we induce this condition to the full state-space $\Omega$. The interpretation of the generalised knowledge $K^{\prime}$ is the same, i.e. $K^{\prime} E \neq \emptyset$ means 'the agent knows $E$'. Definition \ref{d01} satisfies all standard properties of knowledge $K$ and unawareness $U$, and it revises the standard property of necessitation to the set $K^{\prime} \Omega$, so the non-trivial unawareness $UE \neq \emptyset$ does not disappear\footnote{The generalised knowledge $K^{\prime}$ in Definition \ref{d01} is not the same as the awareness operator defined in literature equal to $A(E) := \Omega \setminus U(E)$ \cite{dekel1998,heifetz2006,galanis2013,fukuda2021}, as it has different properties. For example, in a state-space $\{a,b,c \}$ with $P(a) = \{a \}, U(\{a \}) = \{c \}$, we have $A(\{a \}) = \{b,c \} \neq K^{\prime}(\{a \}) = \{a \}$.} \footnote{Current literature asserts that no possibility correspondence can induce a knowledge operator which does not satisfy the standard property of necessitation \cite{dekel1998,heifetz2006,galanis2013,fukuda2021}. The generalised knowledge operator $K^{\prime}$ in Definition \ref{d01} proves that this is not the case.}. 

\begin{theorem} \label{t02}
For all $E \subseteq \Omega$, the generalised knowledge $K^{\prime}$ in Definition \ref{d01} satisfies:

.

$K^{\prime}E = KE, \text{ for all } E \subset \Omega,$

$K^{\prime} \Omega = K\Omega \setminus U \Omega = \Omega \setminus U \Omega.$ (R necessitation)

\end{theorem}

\begin{remark} \label{r02}
For all $E \subseteq \Omega$, $UE \subseteq U \Omega$.
\end{remark}

\begin{proof}
Assume the set partition $\Omega = K^{\prime} \Omega \cup \neg K^{\prime} \Omega \rightarrow K^{\prime} \Omega = \Omega \setminus \neg K^{\prime} \Omega$. From Definition \ref{d01} and Eq. \ref{k01}, we have $\neg K^{\prime} \Omega = U \Omega$. From Theorem \ref{t01}, we have $UE \subseteq U\Omega$.
\end{proof}

\begin{remark} \label{r03}
For all $E \subseteq \Omega$, $U \Omega = \bigcup_{2^{\Omega}} UE$.
\end{remark}

\begin{proof}
From Remark \ref{r02}, $\bigcup_{2^{\Omega}} UE \subseteq U \Omega$ for all $E \subseteq \Omega$. From Theorem \ref{t02}, $\bigcup_{2^{\Omega}} UE = U \Omega$.
\end{proof}

The set $U \Omega$ is the set of states which the agent has non-trivial unawareness. Theorem \ref{t02} derives a revised version of the standard property of necessitation, so called 'R necessitation' equal to $K^{\prime} \Omega = \Omega \setminus U \Omega$. R necessitation is equal to standard necessitation if there is no non-trivial unawareness, i.e. if $UE = \emptyset$ for all events $E \subseteq \Omega$, then $K^{\prime} \Omega = K \Omega = \Omega$.

Definition \ref{d01} satisfies the standard properties of KU introspection equal to $K(UE) = \emptyset$, AU introspection equal to $UE \subseteq U(UE)$, and symmetry\footnote{Property of symmetry for unawareness equal to $UE = U(\neg E)$ is often regarded as a desired property, see \cite{modica1994,heifetz2006}. Nonetheless, if we have $UE \neq \emptyset$, then the traditional definition of unawareness in Eq. \ref{u01} does not satisfy symmetry.} equal to $UE = U(\neg E)$, for all events $E \subset \Omega$, following directly from Theorem \ref{t02}. For the full state-space $\Omega$, Theorem \ref{t02} and Remark \ref{r03} imply that we have $K^{\prime} (U\Omega) = K^{\prime} (\cup_{2^{\Omega}} UE) = K (\cup_{2^{\Omega}} UE) = \emptyset$, so KU introspection obtains. Moreover, Theorem \ref{t02} and KU introspection implies $\neg K^{\prime} (U\Omega) = \Omega$, then $K^{\prime} \neg K^{\prime} (U\Omega) = \Omega \setminus U \Omega$, then $\neg K^{\prime} \neg K^{\prime} (U\Omega) = U \Omega$, and so on. Therefore, applying the definition of unawareness in Eq. \ref{u01} with respect to $K^{\prime}(U\Omega)$ provides the property of AU introspection with respect to the full state-space $\Omega$ equal to

\begin{equation} \label{u02}
\begin{array}{rl}
U(U\Omega) &= \neg K^{\prime} U\Omega \cap \neg K^{\prime} \neg K^{\prime} U\Omega \cap \neg K^{\prime}\neg K^{\prime} \neg K^{\prime} U\Omega \cap \dots \\
\\
&= \Omega \cap U\Omega \cap \Omega \cap U\Omega \cap \Omega \cap \dots = U\Omega . \\
\end{array}
\end{equation}

At last, Theorem \ref{t02} implies $K^{\prime}(\neg \Omega) = K^{\prime} \emptyset = \emptyset$, then $\neg K^{\prime} \emptyset = \Omega$, then $K^{\prime} \neg K^{\prime} \emptyset = K^{\prime} \Omega = \Omega \setminus U \Omega$, then $\neg K^{\prime} \neg K^{\prime} \emptyset = K^{\prime} \Omega = U \Omega$, and so on. Therefore, applying the definition of unawareness in Eq. \ref{u01} with respect to $\neg \Omega$, we have

\begin{equation} \label{u03}
\begin{array}{rl}
U(\neg \Omega) &= \neg K^{\prime} \emptyset \cap \neg K^{\prime} \neg K^{\prime} \emptyset \cap \neg K^{\prime}\neg K^{\prime} \neg K^{\prime} \emptyset \cap \dots \\
\\
&= \Omega \cap U\Omega \cap \Omega \cap U\Omega \cap \Omega \cap \dots = U\Omega , \\
\end{array}
\end{equation}

\noindent which is equal to Eq. \ref{u02}, therefore satisfying the property of symmetry with respect to the full state-space $\Omega$ equal to $U \Omega = U(\neg \Omega)$.

\subsection{Addressing the Dekel-Lipman-Rustichini result} \label{Addressing the Dekel-Lipman-Rustichini result}

The study of \cite{dekel1998} is classical in the literature for it proves contradictory results within the standard state-space model of unawareness using the properties of necessitation, plausibility, KU introspection and AU introspection. In short, \cite{dekel1998} show that an unaware agent satisfying the standard properties of knowledge and unawareness must be indeed aware of everything. \cite{dekel1998} interpret this result as indicating that the state-space model of unawareness is inconsistent\footnote{Curiously, \cite{dekel1998} drop the property of negative introspection while inadvertently invoking necessitation. Theorem \ref{t01} shows that negative introspection leads to necessitation, for both properties are violated if there is non-trivial unawareness.} \cite{dekel1998,heifetz2006,chen2012,galanis2013,schipper2014,fukuda2021}. 

\cite{dekel1998} show that for an agent with non-trivial unawareness $UE \neq \emptyset$, it implies

\begin{equation} \label{dlr}
\begin{array}{rl}
\emptyset \neq UE &\subseteq U(UE) \text{ (AU introspection)} \\
\\
&\subseteq \neg K \neg K (UE) \text{ (plausibility, Eq. \ref{u01})} \\
\\
&= \neg K \Omega \text{ (KU introspection)} \\
\\
&= \emptyset . \text{ (necessitation)} \\
\\
\end{array}
\end{equation}

The last line in Eq. \ref{dlr} derives by applying the property of necessitation, for we have $K\Omega = \Omega \rightarrow \neg K\Omega = \emptyset$. As a result, the contradiction $\emptyset \neq UE = \emptyset$. The generalised knowledge $K^{\prime}$ in Definition \ref{d01} solves this contradiction, for we have 

\begin{equation} \label{dlrR}
\begin{array}{rl}
\emptyset \neq UE &\subseteq U(UE) \subseteq \neg K^{\prime} \neg K^{\prime} (UE) \\
\\
&= \neg K^{\prime} \Omega = U \Omega . \text{ (R necessitation)} \\
\\
\end{array}
\end{equation}

\noindent therefore $UE \subseteq U \Omega$, see Remark \ref{r02}. Overall, the generalised knowledge $K^{\prime}$ preserves the agent's non-trivial unawareness $UE \neq \emptyset$ through the theoretical manipulations of the knowledge operator, therefore preserving model consistency.

\section{Related literature} \label{Related literature}

Literature on unawareness begins with the classical study of \cite{aumann1976} and the notion of 'common knowledge', which refers not only that two agents know an event, but each knows that the other knows it, and each knows that the other knows that she knows it, and so on. For two agents having the same prior beliefs, if their posteriors for an event include common knowledge, then their posteriors become equal \cite{aumann1976}.

The study of \cite{aumann1976} motivates several approaches to model the epistemic states of agents. Some early studies are based on modal logic, for the events are sentences which may be true or false, and modal operators over sentences indicate the knowledge and awareness of the agent \cite{modica1994,modica1999}. Awareness is defined as a concious uncertainty where the agent knows the event, or she knows that she does not know the event\footnote{Awareness derives from negative introspection \cite{modica1994}. i.e. for a sentence $\varphi$ and the modal operator $k \varphi$ meaning 'agent knows $\varphi$', awareness is defined as $a \varphi := k \varphi \vee k \neg k \varphi$. Literature also refers to this approach as Kripke-structures model \cite{schipper2014}.}. These studies propose the now-standard definition of unawareness equal to the negation of awareness, meaning that the agent does not know the event, and she does not know that she does not know the event (and so on). A modal logic of unawareness requires the weakening of some axioms of logic and inference related to monotonicity and the knowledge of tautologies \cite{modica1994,modica1999}.

Further models follow a set-theoretic approach, which became predominant in literature. \cite{bacharach1985} derive one of the first state-space models of knowledge, which is based on a universal set of states of the world, events as subsets of the universal set, and a knowledge operator over states modelling the events known by the agent. \cite{bacharach1985} propose four properties of knowledge to be satisfied by a rational agent, namely the properties of non-delusion, conjunction, positive introspection and negative introspection\footnote{For events $E_{1}, E_{2}$ and a knowledge operator $K$, non-delusion is defined by $KE_{1} \subseteq E_{1}$, conjunction is defined by $K(E_{1} \cap E_{2}) = KE_{1} \cap KE_{2}$, positive introspection is defined by $KE_{1} \subseteq KKE_{1}$.}. Hence, a rational agent obtains a partitional information structure regarding all possible states. The properties proposed by \cite{bacharach1985} became standard in literature, e.g. \cite{bacharach1985,samuelson2004,heifetz2006,li2009,schipper2014}. 

The partitional model of knowledge implies that the agent becomes aware of everything, by monotonicity and negative introspection, so it imposes a limitation on modelling unawareness as the negation of awareness. \cite{geanakoplos1989} is seminal in modelling knowledge of events by means of non-partitional sets of states. The intuition is that at some true state $s$, the agent may consider several states to be possible, and any state may be included in more than one possibility correspondence of the agent. In this approach, the agent knows an event if that event includes all states that are included in her possibility correspondence\footnote{For example, for a state-space $\{a, b \}$ and a possibility correspondence $P(a) = \{a \}, P(b) = \{a, b \}$, the agent knows events $\{ a \}$ and $\{a, b \}$, however she does not know the event $\{ b \}$.}. 

The non-partitional model proposed by \cite{geanakoplos1989} provides the theoretical structure to define unawareness (negation of awareness) within the state-space model of knowledge, i.e. it allows to derive the set of states which the agent is unaware of some event \cite{geanakoplos1989,dekel1998}. It also derives the standard properties of unawareness consistent with the early modal-logic studies of \cite{modica1994,modica1999}, referring to plausibility, KU introspection and AU introspection\footnote{On the analysis of common knowledge and consensus across agents, \cite{geanakoplos1989} propose additional properties, e.g. ref. 'knowing that you know', 'nested', and 'balanced'.} \cite{dekel1998,heifetz2006,li2009,chen2012,schipper2014}. 

Not long after the study of \cite{geanakoplos1989}, \cite{dekel1998} apply the same non-partitional model to show that under the standard properties of knowledge and unawareness, an unaware agent is always aware of everything, see Section \ref{Addressing the Dekel-Lipman-Rustichini result}, i.e. assuming all the standard properties, the unaware agent obtains the same epistemic state as the rational agent with a partitional knowledge. Studies interpret this contradictory result as indicating that the state-space approach is not capable to properly model unawareness, so proposing its full refutation \cite{dekel1998,heifetz2006,li2009,chen2012,galanis2013,schipper2014}. 

Several studies address the result in \cite{dekel1998}. One stream of studies explores further contradictions within the state-space model. A major question is whether the knowledge of an agent is induced by a possibility correspondence \cite{dekel1998,fukuda2021,tada2021}, for a model applying a primitive knowledge operator and axiomatic properties may promote model generalisation \cite{dekel1998,chen2012,schipper2014,fukuda2021,tada2021}. In this axiomatic approach, studies show that not all intuitive properties are consistent with each other, thus some properties must be abandoned. For example, the studies of \cite{chen2012,tada2021} show that by assuming necessitation or symmetry, the properties of negative introspection, KU introspection and AU introspection become operationally equivalent to each other, and especially equivalent to the empty set $\emptyset$. \cite{fukuda2021} show that under a generalised state-space compatible with all properties in the literature, the knowledge of an agent does not always satisfy AU introspection.

A second stream of studies proposes new models intended to solve the result in \cite{dekel1998}. The predominant set-theoretic approach is to model multiple state-spaces and to define knowledge and unawareness with respect to separate subspaces. \cite{heifetz2006} model a lattice of state-spaces which are partially ordered by 'expressiveness' of each subspace, therefore inducing a surjective projection from the more expressive to the less expressive subspace. An agent knows an event with respect to the subspace which she might be aware, and unawareness refers to states in other subspaces which the agent cannot reach by her own introspection\footnote{In detail, each subspace includes mutually exclusive states which are considered possible by the agent, while unaware states are included in 'non-expressible' subspaces. Unaware states are not included in either the set of known events or in their complements \cite{heifetz2006}.}, so called 'non-expressible' subspaces. Non-expressible subspaces imply a new property called 'weak necessitation'\footnote{Formally, assume a universal space $\Omega$, subspaces $S, S^{\prime} \subseteq \Omega$, and the relation $S^{\prime} \succeq S$ meaning '$S^{\prime}$ is at least as expressive as $S$'. Assume that $\cup S^{\prime}$ is the set of all extensions of $S$ to equal or more expressive vocabularies. For an event $E \subseteq S$, weak necessitation implies that the agent is aware of $E$ so $K( \cup S^{\prime}E)$ \cite{dekel1998}.}, which solves the contradictory result in \cite{dekel1998}.

\cite{li2009} model the full state-space as a cartesian product between the set of all possible states which may be true or false and the agent's awareness operator. Events carry information not only about the true states (factual information) but also about other possible states (awareness information), so events may change the less-detailed 'subjective state-space' of the agent\footnote{The set of subjective state-spaces also produce a lattice of state-spaces ordered by expressiveness \cite{li2009}, similar to the model in \cite{heifetz2006}.}. While the agent knows an event $E$ with respect to her subjective state-space only, that event allows the agent to become aware of possible states which she might be unaware before knowing $E$\footnote{For example, for the event 'Bob died', the agent knows that Bob died, and she also becomes aware of several possible states which she was unaware before knowing that Bob died, e.g. there was some creature named Bob, or Bob may be his nickname; Bob could be a person, a pet, a character in a story; there may be an accident, a disease, a crime. Events that contain the same factual information but incorporate more awareness are defined as 'elaborations' of an event \cite{li2009}.}. Subjective state-spaces imply the property of 'subjective necessitation' which is equivalent to necessitation with respect to the agent's subjective state-space only, therefore proposing a solution for the result in \cite{dekel1998}.

\cite{fukuda2021} recently provide a rigorous generalised model which combines the standard state-space model by \cite{geanakoplos1989,dekel1998} and the multiple state-spaces by \cite{heifetz2006}. The generalised state-space is defined as a triple including a complete algebra, a collection of partial orders on subsets and a collection of surjective projections on subsets as in \cite{heifetz2006}. An event $E$ is defined within the base subspace which it is expressed, so the knowledge of an agent depends on the subspace which she is aware. The generalised model fully replicates the properties of the lattice model by \cite{heifetz2006}, however it does not satisfy AU introspection when applied on the standard non-partitional model\footnote{\cite{fukuda2021}, Section 3.1 shows a detailed analysis of all standard properties of knowledge and unawareness following the application of the general model on the standard non-partitional approach.}. Inspection of the model shows that this limitation derives from the property of necessitation inherited from the standard state-space approach. 

\section{Discussion} \label{Discussion}

There are two relevant observations. First, necessitation is a partitional property. It derives from early decision models where the knowledge of agents partitions the full state-space, therefore inducing a belief measure on events \cite{aumann1976,samuelson2004}. The application of the unawareness operator in Eq. \ref{u01} for some event $E$ often leads to the expression $K\Omega$. It would be rather spontaneous to apply the standard definition of knowledge in Eq. \ref{k01} hence invoking necessitation, $K\Omega = \Omega$. From this stage and on, further iterations of the knowledge operator keep jumping between the full state-space $\Omega$ and the empty set $\emptyset$, so any unawareness vanishes\footnote{We refer to iterations of the form

\begin{equation*}
\begin{array}{ccc}
\begin{array}{rl}
KE &= \{s \} , \\
\neg KE &= \{s^{\prime} \} , \\
K \neg KE &= K(\{s^{\prime} \}) = \emptyset , \\
\neg K \neg KE &= \Omega , \\
K \neg K \neg KE &= K\Omega , \\
\end{array} &\text{or}&
\begin{array}{rl}
KE^{\prime} &= \emptyset , \\
\neg KE^{\prime} &= \Omega , \\
K \neg KE^{\prime} &= K\Omega , \\
\end{array}
\end{array}
\end{equation*}

\noindent for states $s, s^{\prime} \in \Omega$ and events $E, E^{\prime} \subseteq \Omega$. For the unawareness operator in Eq. \ref{u01}, necessitation implies $UE = \Omega \cap \emptyset \cap \dots = \emptyset$.}. This is clearly unwanted for a consistent model of unawareness. Models on multiple state-spaces are successful in preserving non-trivial unawareness precisely by deriving more restrictive versions of necessitation, e.g. weak necessitation in \cite{heifetz2006}, and subjective necessitation in \cite{li2009}. The property of R necessitation in Theorem \ref{t02} preserves non-trivial unawareness within the standard state-space model.

Second, the knowledge $KE$ does not imply that the agent knows all states within that event, for we may have $E \setminus KE \neq \emptyset$. Likewise, the full state-space $\Omega$ includes states which the agent is unaware of some events. The standard knowledge operator implies $K \Omega = \Omega$, for it does not reflect this unawareness structure, especially when $\Omega$ is infinite. An unaware agent does not know what she does not know, although she knows that there are things she does not know yet \cite{modica1994,dekel1998}. Definition \ref{d01} derives the unawareness set $U \Omega$, such that for a learning agent that refines her possibility correspondence and becomes aware of new events, it implies $s \in K^{\prime} \Omega \rightarrow s \notin U \Omega$. Now, if there is always something new to learn, we may assume an infinite unawareness set so we have $|U \Omega| = \infty, |K^{\prime} \Omega| < \infty, \Omega = K^{\prime} \Omega \cup U \Omega \rightarrow |\Omega| = \infty$\footnote{$|\Omega|$ refers to the number of elements in the set $\Omega$.}. The generalised knowledge $K^{\prime}$ in Definition \ref{d01} provides the appropriate set-theoretic representation of non-trivial unawareness, including the model of infinite state-space.


\begin{thebibliography}{50} % 50 is a random guess about the number of references
\bibitem{aumann1976}
Aumann, R. A. (1976). Agreeing to disagree. \emph{The Annals of Statistics, 4}, 1236-1239
\bibitem{bacharach1985}
Bacharach, M. (1985). Some extensions of a claim of Aumann in an axiomatic model of knowledge. \emph{Journal of Economic Theory, 37}(1), 167-190.
\bibitem{chen2012}
Chen, Y. C., Ely, J. C., \& Luo, X. (2012). Note on unawareness: Negative introspection versus AU introspection (and KU introspection). \emph{International Journal of Game Theory, 41}, 325-329.
\bibitem{dekel1998}
Dekel, E., Lipman, B. L., \& Rustichini, A. (1998). Standard state-space models preclude unawareness. \emph{Econometrica, 66}(1), 159-173.
\bibitem{fukuda2021}
Fukuda, S. (2021). Unawareness without AU Introspection. \emph{Journal of Mathematical Economics, 94}, 102456.
\bibitem{galanis2013}
Galanis, S. (2013). Unawareness of theorems. \emph{Economic Theory, 52}, 41-73.
\bibitem{geanakoplos1989}
Geanakoplos, J. (1989). Game theory without partitions, and applications to speculation and consensus, Cowles Foundation Discussion Paper no. 914.
\bibitem{heifetz2006}
Heifetz, A., Meier, M., \& Schipper, B. C. (2006). Interactive unawareness. \emph{Journal of Economic Theory, 130}(1), 78-94.
\bibitem{li2009}
Li, J. (2009). Information structures with unawareness. \emph{Journal of Economic Theory, 144}(3), 977-993.
\bibitem{modica1994}
Modica, S., \& Rustichini, A. (1994). Awareness and partitional information structures. \emph{Theory and Decision, 37}, 107-124.
\bibitem{modica1999}
Modica, S., \& Rustichini, A. (1999). Unawareness and partitional information structures. Games and economic Behavior, 27(2), 265-298.
\bibitem{samuelson2004}
Samuelson, L. (2004). Modeling knowledge in economic analysis. \emph{Journal of Economic Literature, 42}(2), 367-403.
\bibitem{schipper2014}
Schipper, B. C. (2014). Unawareness — a gentle introduction to both the literature and the special issue. \emph{Mathematical Social Sciences, 70}, 1-9.
\bibitem{tada2021}
Tada, Y. (2021). Note: AU Introspection and Symmetry under Non-Trivial Unawareness. \emph{SSRN Working paper 3906714}.


\end{thebibliography}
\end{document}